\documentclass{sig-alternate-2013}

\usepackage{calc}
\usepackage{amstext}
\usepackage{multicol}
\usepackage{pslatex}
\usepackage[font=small,labelfont=bf,justification=justified,width=.45\textwidth]{caption}
\usepackage{subcaption}

\usepackage{pgfplots}

\usepackage[utf8]{inputenc}

\usepackage{times}
\usepackage{graphicx}

\usepackage{amsmath,amsfonts,mathrsfs,amssymb,mathtools}
\usepackage{stmaryrd}
\usepackage{color}
\usepackage{url}

\usepackage{multirow}

\usepackage{algorithm,algorithmicx,float}
\usepackage[noend]{algpseudocode}
\DeclareCaptionType{copyrightbox}

  \newdef{definition}{Definition}
  \newtheorem{theorem}{Theorem}

\newsavebox{\fmbox}

\makeatletter
\let\@copyrightspace\relax
\makeatother

\begin{document}

\title{On the Incremental Asymmetric Signatures}

\numberofauthors{1}
\author{
\alignauthor
Kevin Atighehchi\\
       \affaddr{Aix-Marseille University, CNRS, Centrale Marseille, ERISCS, I2M, UMR 7373, 13453 Marseille, France}\\
       \email{kevin.atighehchi@univ-amu.fr}
}

%
%
%
%
%
%
%

\maketitle


\begin{abstract}
The purpose of incremental cryptography is to allow the updating of cryptographic forms of documents undergoing modifications, 
more efficiently than if we had to recompute them from scratch.
This paper defines a framework for using securely a variant of the incremental hash function designed by Bok-Min Goi et \emph{al}. 
The condition of use of their hash function is somehow impractical since they assume that the blocks of the message are all distinct. 
In this paper we show how we can discard this strong assumption so as to construct 
the first practical incremental asymmetric signature scheme that keeps efficient update operations.
Finally, as the proposed scheme has the defect to severely expand the signature size, we propose a solution which drastically reduces this drawback.
\end{abstract}







\section{Introduction}

\begin{sloppypar}
Incremental cryptography, introduced by Bellare, Goldreich and Goldwasser
in \cite{Bellare94incrementalcryptography:,Bellare95incrementalcryptography,Bellare97anew}, 
is used to maintain up-to-date at low computational cost the outputs of cryptographic algorithms.
If one has the signature of the current version of a file, it is preferable to avoid to recompute 
from scratch the signature algorithm applied to the entire file whenever a modification is performed. Such a low computational
efficiency for update operations finds application in various situations, for example when we want to maintain constantly 
changing authenticated databases, editable documents or even when we have to sign similar letters addressed to different recipients.
Nowadays, an area of application of greatest interest is its use in cloud storage systems, trends to outsource personal data make
computational efficiency as well as space efficiency particularly important.
Besides, an other interest of incremental cryptographic schemes is their inherent parallelism which allow performances to fit accordingly 
by using appropriate multi-processing platforms.
\end{sloppypar}

Typical operations supported by incremental cryptography are the \emph{replace}, \emph{insert} and \emph{delete} operations. 
Authentication schemes based on static hash tree can support only both replacement and appending of a block in $\mathcal{O}(\log n)$ where 
$n$ is the number of blocks of the message. 
Such schemes are used to authenticate fixed-size memory \cite{DBLP:journals/tcos/ElbazCGLPT09}. Those based on a dynamic hash tree (or skip list) 
\cite{Bellare95incrementalcryptography,Micciancio97obliviousdata,Goodrich01efficientauthenticated}  can support additionally the insertion/deletion of a block 
in $\mathcal{O}(\log n)$ and can be used for the authentication of variable-sized data, such as files.
The running time of an insert operation can be proportional with the amount of changes made in $D$ to obtain $D'$
if they are based on weakly related data structures such as pair-wise chainings. This is the case for the PCIHF (for ``Pair Chaining \& Modular Arithmetic Combining Incremental Hash Function'') 
hash function defined in \cite{DBLP:conf/indocrypt/GoiSC01,Computational_complexity_PCIHF} or for one of the MAC (Message Authentication Code) algorithms defined 
in \cite{Bellare95incrementalcryptography}. 
Note that randomized self-balanced data structures (such as the signatures based on Oblivious Trees \cite{Micciancio97obliviousdata} or the 
authenticated dictionary \cite{Goodrich01efficientauthenticated} based on non-deterministic skip lists) and pair-wise chaining ensure obliviousness at each modification 
on the signature (resp. tag), that is, the updated signature (resp. tag) is indistinguishable from an original one.

The problem of pair-wise chaining is that input blocks are weakly related and by consequent the PCIHF hash function is subject to trivial collisions 
\cite{pairchaining,pair-wise-chaining-hash}. 
The use of the \emph{randomize-then-combine} paradigm in hash functions raises security issues due to the combining operator used.
By simply performing a Gaussian elimination to exhibit linear dependencies, Bellare et \emph{al} \cite{Bellare97anew} have proved that the XOR combining 
operator is also subject to simple collisions when used in un-keyed hash function. 
This one is discarded in the PCIHF hash function in favour of a modular arithmetic operator to avoid this second type of 
attack. Unfortunately, the use of such an operator remains delicate to recent sub-exponential attacks \cite{Wagner02ageneralized,Ber07}.

Bok-Min Goi et \emph{al} \cite{DBLP:conf/indocrypt/GoiSC01} have designed an efficient 
incremental hash function based on the \emph{randomize-then-combine} paradigm \cite{Bellare97anew}.
Unfortunately, their construction is based on the impractical assumption that the blocks
in the message to hash are all distinct. Besides, due to recent attacks described in the Generalized Birthday Problem \cite{Wagner02ageneralized,Ber07},
the security parameters of the PCIHF hash function are no longer adapted.
We show in this paper how one can discard this strong assumption about the blocks of a message to construct 
the first practical incremental asymmetric signature scheme that keeps update operations within a linear runtime 
worst case complexity. Since the first incremental signature scheme that we describe produces signatures of size 
twice the length of the messages, we propose a parametrizable version which reduces this overhead with the counterpart 
of less time efficient updates.
Finally, we show that we can reduce the security of our schemes to the security of the underlying primitives 
involved, namely a set-collision resistant hash function, a traditional signature scheme and a simple combinatorial problem.

The paper is organized as follows: in the next section, 
we discuss pair-wise chaining and give some useful definitions.
In Section \ref{proposed_scheme}, we introduce our incremental signature scheme for which an improvement is given in Section \ref{improvment}. 
Proofs of security are given in Section \ref{security} and 
efficiency of our solution is discussed in Section \ref{efficiency}.
Finally we conclude in Section \ref{conclude}.


\section{Preliminaries}\label{prel}

Cryptographic schemes take as input a document $D$ which is divided into a sequence of fixed-size blocks $D_1,D_2, \ldots,D_n$. 
Documents are then viewed as strings over an alphabet $\sum = \{0, 1\}^{b}$ where $b$ is the block-size in bits.

\subsection{Operations on documents} 

\begin{sloppypar}
We denote by $\mathcal{M}$ the space of modifications
and define the main modification operations $M \in \mathcal{M}$ allowed on a document as follows:
($i$) $M=(delete,i,j)$ deletes data from the block $i$ to $j$ (included);
($ii$) $M=(insert,i,\sigma)$ inserts $\sigma$ between the $i^{\textrm{th}}$ and $(i+1)^{\textrm{th}}$ block, where $\sigma$ represents data of size a multiple of $b$;
($iii$) $M=(replace,i, \sigma)$ changes to $\sigma$ the data starting from block $i$ to block $(i+k-1)$ (included), where $\sigma$ is of size $kb$.

The resulting document after a modification operation $M$ is denoted $D \langle M \rangle$. Other frequent modifications are the $cut$ and $paste$ operations but, 
for the sake of simplicity, we do not deal with such composite modifications here.
\end{sloppypar}

\subsection{Hash functions based on pair-block chaining}

\begin{sloppypar}
\textit{Pair-wise chaining}\footnote{The term \textit{pair-block chaining} is also encountered.} roughly describes the following process: given the $n$-block string $D$, 
each block (except the last one) is paired up with the subsequent one, a pseudo-random function is then evaluated at the resulting point. We obtain in this way a sequence
of $n-1$ values which characterizes a relation between the blocks of the document. The following paragraph gives a detailed description.
\end{sloppypar}

\begin{sloppypar}
\paragraph*{\textbf{The PCIHF hash function}} This is the first really incremental hash function \cite{DBLP:conf/indocrypt/GoiSC01} implementing 
the \emph{randomize-then-combine} paradigm introduced by Bellare et \emph{al} \cite{Bellare97anew}, as this one supports \emph{insert} and \emph{delete} operations. 
A document $D$ is divided into a sequence of blocks of size $b$. If the 
document length is not a multiple of $b$-bits, a standard padding method is applied at its tail with a bit ``1'' followed by the sufficient 
number of consecutive bits ``0'' so that the last block is of size $b$. The final message is of the form $D=D_1\|D_2\|\ldots\|D_n$ with $D_i \in \{0,1\}^b$ for 
$i \in \llbracket 1,n \rrbracket$. Besides, a stronger assumption is made about the blocks of the document which have to be all distinct,
meaning that $D_i\neq D_j$ for $i, j \in \llbracket 1,n \rrbracket$ and $i \neq j$. The hash value $\mu$ for PCIHF is calculated as following:
\end{sloppypar}

$$\mu = \sum_{i=1}^{n-1} \mathcal{R}(D_i\|D_{i+1}) \mod{2^{160}}$$
where the \emph{randomize} operation $\mathcal{R}:\{0,1\}^{2b} \rightarrow \{0,1\}^{160}$ is a compression function, or a pseudo-random function. Note that 
the \break padding used is not secure, as shown with examples in \cite{pairchaining}. Besides, the authors of PCIHF advocate the use of a standard hash function for 
the \emph{randomize} operation and in this case the padding is not required. This is because in standard hash functions a padding is 
enforced whatever the input length is.

To ensure integrity of content a hash function must ensure first and second pre-image resistance as well as collision resistance. The security of a signature scheme based on 
the \emph{hash-then-sign} paradigm relies on the collision resistance property of the underlying hash function. A hash function is said to be collision resistant if it 
is hard to find two messages that hash to the same output value. The PCIHF hash function allows insert operation with a 
running time proportional to the number of blocks to insert, and replace/delete operations with constant cost. Besides, this one has been proved to be
collision resistant in the random oracle model \cite{DBLP:conf/indocrypt/GoiSC01}.

\paragraph*{\textbf{The design problem}}

Let us denote $(X,Y)$ a subsequence of blocks starting with $X$ and ending with $Y$, and $[X,Y]$ a pair of consecutive blocks. 
We recall here the problems raised in \cite{pairchaining} about pair-wise chaining used in hash functions. The authors found message patterns which lead to trivial collisions:
($i$) Palindromic messages of any block length, for instance the 3-block message $B\|A\|B$ which produces the same hash value than $A\|B\|A$;
($ii$) Certain non-palindromic messages. As said explicitly in \cite{pairchaining}, "any two messages with the same block at both ends, 
and where all consecutively paired blocks follow the same order, would cause collisions". The given examples are the messages $A\|B\|C\|B\|A$ 
and $B\|C\|B\|A\|B$ which produce the same hash value;
($iii$) In the case where the XOR operator is used, messages with repetitive blocks could also produce collisions. This is the case of the messages $A\|B\|B\|B\|C$ and
$A\|B\|C$ for which a pair number of pair-wise links $[B,B]$ does not change the hash value.
\begin{sloppypar}
Now, concerning the second case above we can observe the following two more general patterns leading to collisions: 
($i$)~If a message contains at least three identical blocks, for instance $B\|X\|C\|X\|D\|X\|E$, then we have the interesting sequence of
pair-wise links $[X,C]$, $[C,X]$, $[X,D]$ and $[D,X]$. We see that we can permute the subsequences ($[X,C]$, $[C,X]$) and ($[X,D]$, $[D,X]$) 
so that the message $B\|X\|D\|X\|C\|X\|E$ produces the same hash value;
($ii$)~If a message contains at least two pairs of identical blocks, for instance the message $B\|\underbrace{X\|C\|Y}\|D\|\underbrace{X\|E\|Y}\|F$ for which 
the two subsequences $(X,X)$ and $(Y,Y)$ overlap, a permutation gives us the message $B\|X\|E\|Y\|D\|X\|C\|Y\|F$ which produces the same hash value.
\end{sloppypar}

\subsection{Incremental signature schemes}

\begin{definition}
An incremental asymmetric signature scheme is specified by a 4-tuple of algorithms $\Pi=(\mathcal{G},\mathcal{S},\mathcal{I},\mathcal{V})$ in which:
\begin{itemize}
 \item $\mathcal{G}$, the key generation algorithm, is a probabilistic polynomial time algorithm
that takes as input a security parameter $k$ and returns a key pair $(sk,pk)$ where $sk$ is the private key and 
$pk$ the public key.
\item $\mathcal{S}$, the signature algorithm, is a probabilistic polynomial time algorithm
that takes as input $sk$ and a document $D \in \sum^+$ and returns the signature $s=\mathcal{S}_{sk}(D)$ where $s$ is a 
signature with appendix,  that is, $s$ is of the form $(D,s')$ where $D$ is the document and $s'$ the appendix.
\item $\mathcal{I}$, the incremental update algorithm, is a probabilistic polynomial time algorithm 
that takes as input a key $sk$, (a document $D$), a modification operation $M \in \mathcal{M}$, and the
signature $s$ (related to $D$) and returns the modified signature $s'$.
\item $\mathcal{V}$, the verification algorithm, is a deterministic polynomial time algorithm
that takes as input a public key $pk$ and a signature $s=\mathcal{S}_{sk}(D)$ and returns $1$ if the signature is valid, 
$0$ otherwise.
\end{itemize} 
\end{definition}

Considering a modified document $D'=D \langle M \rangle$, the desired behaviours of an incremental signature scheme are the followings:
($i$)~It is required that $\mathcal{V}_{pk}(\mathcal{I}_{sk}(\mathcal{S}_{sk}(D),M))=1$.
($ii$)~Optionally, an incremental signature scheme could be oblivious (or perfectly private) in the sense that the output of a signature $S_{sk}(D')$ is indistinguishable from
the ouput of an incremental update $\mathcal{I}_{sk}(\mathcal{S}_{sk}(D),M)$. This property is particularly useful if we want to hide the modification history
of a signed document, or even the fact that an update operation has been performed.


\subsection{Unforgeability}\label{Unf_game}

For incremental signature schemes, existential unforgeability measures the unability for an adversary to generate a new pair $(D^*, S^*)$ where:
($i$) $D^*$ is not a document that has been signed by the signing oracle;
($ii$) $D^*$ is not a modified document obtained via the incremental update oracle;
($iii$) $S^*$ is a signature on $D^*$.
More precisely, the notion of 
\emph{existential unforgeability} under a chosen message attack is defined using the following game between the adversary $\mathcal{A}$ and the challenger:
\begin{enumerate}
 \item The challenger runs algorithm $\mathcal{G}$ to obtain a public key $pk$ and a private key $sk$. The adversary $\mathcal{A}$ is given $pk$.
 \item $\mathcal{A}$ chooses and requests signatures (adaptively) on at most $q_s$ messages. Additionally, $\mathcal{A}$ chooses and requests 
at most $q_i$ valid incremental updates (adaptively) on signatures issued from the signing oracle (or from the updating oracle). The term ``valid`` means that 
the pair $(D,S)$ on which the update is requested satisfies $\Pi.\mathcal{V}_{pk}(D,S)=1$.  The challenger responds to each query.
 \item Eventually, $\mathcal{A}$ outputs a pair $(D^*,S^*)$ where $D^*$ is not a document signed by the signing oracle $\mathcal{O}^{S_{sk}}$ nor
an updated document whose the signature has been obtained from the incremental update oracle $\mathcal{O}^{\mathcal{I}_{sk}}$. We say that $\mathcal{A}$
succeeds if $\Pi.\mathcal{V}_{pk}(D^*,S^*)=1$.
\end{enumerate}

\begin{definition}
Let $\Pi = (\mathcal{G}, \mathcal{S}, \mathcal{I}, \mathcal{V})$ be an incremental signature scheme over modification space $\mathcal{M}$, and let $\mathcal{A}$ 
be an adversary. 
Let
$$Adv^{unf}_{\mathcal{A},\Pi} = Pr[sk \leftarrow G; S^* \leftarrow \mathcal{A}^{\mathcal{S}_{sk},\mathcal{I}_ {sk}} : \mathcal{V}_{pk}(S^*) = 1].$$

We say that $\Pi$ is $(t,q_s,q_{i},\epsilon)$-secure in the sense of \emph{existential unforgeability} if, for any adversary $\mathcal{A}$ which runs in time $t$, making
$q_s$ queries to the signing $\mathcal{O}^{S_{K'}}$ oracle and $q_{i}$ valid queries to the update $\mathcal{O}^{I_{K'}}$ oracle, 
$Adv^{unf}_{A,\Pi}$ is less than $\epsilon$.
\end{definition}


\section{The proposed incremental asymmetric signature scheme}\label{proposed_scheme}

We assume the use of a simpler version of the PCIHF hash function, denoted $H$, where input blocks of the random function are expressly unrelated. 
This incremental hash function has been proved to be set-collision resistant \cite{conf/asiacrypt/ClarkeDDGS03} in the random oracle model under the weighted knapsack assumption.
From there we will use it as a black-box primitive 
to design secure incremental signature schemes. First we format the message correctly by enforcing the padding in the following way:
if the message length is not a multiple of $2b$-bit, we padd the last block with a bit 1 followed by the sufficient number of bits 0, otherwise we add a new block
whose content is $1\{0\}^{2b-1}$. The final message is of the form $M=M_1\|M_2\|\ldots\|M_{n-1}$ with $M_i \in \{0,1\}^{2b}$ for 
$i \in \llbracket 1,n-1 \rrbracket$. The hash value is computed such that $\mu=H(D)$ where $H(.)$ is a function taking a string of size a multiple of $2b$-bit 
and returning a value in $\{0,1\}^{3200}$. This last one is defined as follows:
$$H(M) = \sum_{i=1}^{n-1} \mathcal{R}(M_i) \mod{2^{3200}}$$
where the randomize function $\mathcal{R}:\{0,1\}^{2b} \rightarrow \{0,1\}^{3200}$ has an output size of 3200 bits and the addition is performed modulo $2^{3200}$. 
Considering the problem of generating a second preimage, the output size of $\mathcal{R}$ in the original PCIHF hash function is no longer secure due to the efficient 
(subexponential) generalized birthday attack \cite{Wagner02ageneralized,Ber07}. 
This force us to increase the output size of $\mathcal{R}$ by considering for instance 3200 bits, a choice of parameter which allows us to guard against this attack by 
keeping an equivalent security of 112 bits.

Let us suppose a correctly padded $(n-1)$-block message $D=D_1\|D_2\|\ldots\|D_{n-1}$ with $D_i \in \{0,1\}^{b}$ for $i \in \llbracket 1,n-1 \rrbracket$. 
We assume the use of a digital signature algorithm $\Psi=(G,S,V)$ in which: ($i$) $G$ takes as input a security parameter and returns a pair of keys $(sk,pk)$; 
($ii$) $S$ is a \emph{ppt} signature algorithm taking as input the secret key $sk$, a document $D$ and returns a signature $s$;
($iii$) $V$ is a \emph{dpt} algorithm taking as input the public key $pk$, a document $D$, a signature $s$ and returns 1 if the signature is valid,
$\bot$ otherwise.

\begin{sloppypar}
Now we can describe the 4-tuple of algorithm $\mathrm{IncSIG}=(\mathcal{G},\mathcal{S},\mathcal{I},\mathcal{V})$. 
The key generation of the incremental algorithm is simply $\mathrm{IncSIG}.\mathcal{G}=\Psi.G$. 
The signature algorithm $\mathrm{IncSIG}.\mathcal{S}$ taking as input the secret key $sk$ and the document $D$ is the following:
\end{sloppypar}

~\\
\linethickness{3pt}
\centerline{\fbox{\begin{minipage}{\linewidth}
\vglue 2mm
\sf
\begin{enumerate}
 \item Pick uniformly at random $n$ blocks of size $b/2$ denoted $R_{1}$, $R_{2}$, ..., $R_{n}$;
 \item Compute the hash value 
\begin{eqnarray*}
 \mu = \sum_{i=1}^{n-1} \mathcal{R}(R_{i}\|R_{i+1}\|D_i) \mod{2^{3200}};
\end{eqnarray*}
 \item Let $l$ be the size of $D$, compute a signature $s=\Psi.S_{sk}(\mu\|l)$ and return the incremental signature $(R_{1}, R_{2}, \ldots, R_{n}, \mu, s)$.
\end{enumerate}
\vglue 2mm
\end{minipage}}
}
~\\

The verification algorithm $\mathrm{IncSIG}.\mathcal{V}$ taking as input the public key $pk$, the document $D$ and the signature $(R'_{1}, R'_{2}, \ldots, R'_{n}, \mu', s')$ 
is the following:

~\\
\linethickness{3pt}
\centerline{\fbox{\begin{minipage}{\linewidth}
\vglue 2mm
\sf
\begin{enumerate}
 \item Compute the hash value 
\begin{eqnarray*}
\mu = \sum_{i=1}^{n-1} \mathcal{R}(R'_{i}\|R'_{i+1}\|D_i) \mod{2^{3200}};
\end{eqnarray*}
 \item Let $l$ be the size of $D$, run the verification $b=\Psi.V_{pk}(s',\mu\|l)$ and return $b$.
\end{enumerate}
\vglue 2mm
\end{minipage}}
}
~\\

\begin{sloppypar}
We describe hereafter the incremental algorithm $\mathrm{IncSIG}.\mathcal{I}$ that takes as input the secret key $sk$, the document $D$, the signature $(R_{1}, R_{2}, \ldots, R_{n}, \mu, s)$ and 
an insert operation $M=(insert,i,\sigma)$ that changes $D$ in $D'$ where $\sigma$ is only one block (for the sake of simplicity):
\end{sloppypar}

~\\
\linethickness{3pt}
\centerline{\fbox{\begin{minipage}{\linewidth}
\vglue 2mm
\sf
\begin{enumerate}
 \item Draw a new random block of size $b/2$ denoted $R$;
 \item Compute the hash value 
\begin{displaymath}
   \begin{aligned}
 \mu' = \mu & -  \mathcal{R}(R_{i}\|R_{i+1}\|D_i))\\
& + \mathcal{R}(R_{i}\|R\|D_i)\\
& + \mathcal{R}(R\|R_{i+1}\|\sigma) \mod{2^{3200}};
   \end{aligned}
\end{displaymath}
 \item Let $l'$ be the size of $D'$, compute a signature $s'=\Psi.S_{sk}(\mu'\|l')$ and return the updated incremental 
signature $(R_{1}, R_{2}, \ldots, R_i, R, R_{i+1}, \ldots, R_{n}, \mu', s')$.
\end{enumerate}
\vglue 2mm
\end{minipage}}
}
~\\

We do not describe the other operations which can be deduced from the previous one. 


\paragraph*{Remark} We could change the verification algorithm and reject signatures with non distinct $R_i$'s. Checking that all the random 
blocks are distinct would simplify the security analysis, but there are several reasons for not doing this:
($i$)  This verification is a costly operation; ($ii$) The owner of the secret key is not considered as an adversary in the standard definition 
of existential unforgeability. The signing oracle is then implemented exactly as specified.


\section{A simple improvement IncSIG*}\label{improvment}

By considering the same previous notations, let us also define two fixed integers $(k,d)$ $\in {\mathbb{N}^*}^2$ such that $b=kd$ with $d \ge 2$.
We describe an improvement of $\mathrm{IncSIG}$ in which we use a $d$-wise chaining of random $k$-bit blocks. This parametrization will allow
a user to find a time/space efficiency trade-off without sacrifying security.

\begin{sloppypar}
Now we describe the 4-tuple of algorithm $\mathrm{IncSIG}^*=(\mathcal{G},\mathcal{S},\mathcal{I},\mathcal{V})$. 
The key generation of the incremental algorithm is simply $\mathrm{IncSIG}^*.\mathcal{G}=\Psi.G$. 
The signature algorithm $\mathrm{IncSIG^*}.\mathcal{V}$ taking as input the secret key $sk$ and the document $D$ is the following:
\end{sloppypar}

~\\
\linethickness{3pt}
\centerline{\fbox{\begin{minipage}{\linewidth}
\vglue 2mm
\sf
\begin{enumerate}
 \item Pick uniformly at random $n$ blocks of size $\frac{b}{d}$ denoted $R_{1}$, $R_{2}$, ..., $R_{n}$, $R_{n+1}$, ..., $R_{n+d-2}$;
 \item Compute the hash value 
\begin{eqnarray*}
 \hspace{-3 mm} \mu = \sum_{i=1}^{n-1} \mathcal{R}(R_{i}\|R_{i+1}\| \ldots \| R_{i+d-2} \| R_{i+d-1} \| D_i) \mod{2^{3200}};
\end{eqnarray*}
 \item Let $l$ be the size of $D$, compute a signature $s=\Psi.S_{sk}(\mu\|l)$ and return the incremental signature $(R_{1}, R_{2}, \ldots, R_{n+d-2}, \mu, s)$.
\end{enumerate}
\vglue 2mm
\end{minipage}}
}
~\\

The verification algorithm $\mathrm{IncSIG^*}.\mathcal{V}$ taking as input the public key $pk$, the document $D$ and the signature $(R'_{1}, R'_{2}, \ldots, R'_{n+d-2}, \mu', s')$ 
is the following:

~\\
\linethickness{3pt}
\centerline{\fbox{\begin{minipage}{\linewidth}
\vglue 2mm
\sf
\begin{enumerate}
 \item Compute the hash value 
\begin{eqnarray*}
\hspace{-3 mm} \mu = \sum_{i=1}^{n-1} \mathcal{R}(R'_{i}\|R'_{i+1}\| \ldots \| R'_{i+d-2} \| R'_{i+d-1} \| D_i) \mod{2^{3200}};
\end{eqnarray*}
 \item Let $l$ be the size of $D$, run the verification $b=\Psi.V_{pk}(s',\mu\|l)$ and return $b$.
\end{enumerate}
\vglue 2mm
\end{minipage}}
}
~\\

We describe hereafter the incremental algorithm $\mathrm{IncSIG^*}.\mathcal{I}$ that takes as input the secret key $sk$, the document $D$, the signature 
$(R_{1}, R_{2}, \ldots, R_{n+d-2}, \mu, s)$ and 
an insert operation $M=(insert,i,\sigma)$ that changes $D$ in $D'$ where $\sigma$ is only one block (for the sake of simplicity):

~\\
\linethickness{3pt}
\centerline{\fbox{\begin{minipage}{\linewidth}
\vglue 2mm
\sf
\begin{enumerate}
 \item Draw a new random block of size $\frac{b}{d}$ denoted $R$;
 \item Compute the hash value 
\begin{displaymath}
   \begin{aligned}
 \mu' & =  \mu - \sum_{j=i-d+2}^{i} \mathcal{R}(R_{j}\|R_{j+1}\| \ldots \| R_{j+d-1} \| D_j)\\
& + \mathcal{R}(R_{i-d+2} \| \ldots \| R_{i} \| R \| D_{i-d+2})\\
& + \mathcal{R}(R_{i-d+3} \| \ldots \| R_{i} \| R \| R_{i+1} \| D_{i-d+3})\\
& + \sum_{j=i-d+4}^{i-1} \mathcal{R}(R_{j}\| \ldots \| R_i \| R\| R_{i+1} \| \ldots \| R_{j+d-2} \| D_j)\\
& + \mathcal{R}(R_i \| R\|R_{i+1}\| \ldots \| R_{i+d-2} \| D_i)\\
& + \mathcal{R}(R\|R_{i+1}\| \ldots \| R_{i+d-1} \| \sigma) \mod{2^{3200}};
   \end{aligned}
\end{displaymath}
 \item Let $l'$ be the size of $D'$, compute a signature $s'=\Psi.S_{sk}(\mu'\|l')$ and return the updated incremental 
signature $(R_{1}, \ldots, R_i, R, R_{i+1}, \ldots, R_{n+d-2}, \mu', s')$.
\end{enumerate}
\vglue 2mm
\end{minipage}}
}
~\\

Let us now describe this incremental algorithm in the case of
a replace operation $M=(replace,i,\sigma)$ where $\sigma$ is only one block (for the sake of simplicity):

~\\
\linethickness{3pt}
\centerline{\fbox{\begin{minipage}{\linewidth}
\vglue 2mm
\sf
\begin{enumerate}
 \item Draw a new random block of size $\frac{b}{d}$ denoted $R$;
 \item Compute the hash value 
\begin{displaymath}
   \begin{aligned}
 \mu' = \mu & - \mathcal{R}(R_{i} \| \ldots \| R_{i+d-1} \| D_{i})\\
& + \mathcal{R}(R\|R_{i+1}\| \ldots \| R_{i+d-1} \| \sigma) \mod{2^{3200}};
   \end{aligned}
\end{displaymath}
 \item Let $l'$ be the size of $D'$, compute a signature $s'=\Psi.S_{sk}(\mu'\|l')$ and return the updated incremental 
signature $(R_{1}, \ldots, R_i, R, R_{i+1}, \ldots, R_{n+d-2}, \mu', s')$.
\end{enumerate}
\vglue 2mm
\end{minipage}}
}
~\\

When performing a deletion operation, in order to maintain consistency in the $d$-wise chain, the contribution of $d$ hash values of the non updated chain 
has to be deducted from $\mu$ while $d-1$ new values have to be added. This algorithm can be deduced from the insertion one.

\paragraph*{Remark} We could change the verification algorithm and reject signatures with non distinct $(d-1)$-tuples of random blocks. Checking that all these $(d-1)$-tuples 
are distinct would simplify the security analysis, but there are several reasons for not doing this:
($i$)  This verification is a costly operation; ($ii$) The signer is not considered as an adversary in the standard definition of existential unforgeability.

\section{Security analysis}\label{security}

The property of obliviousness of the schemes $\mathrm{IncSIG}$ and $\mathrm{IncSIG^*}$ is obvious and we focus only on the security analysis of unforgeability.

\begin{sloppypar}
\begin{theorem}
Let suppose that $\mathcal{F}$ is a $(t,q_i,q_s,\epsilon)$-forger against our incremental signature scheme $\mathrm{IncSIG}$, then there exists a $(t',q_h,\epsilon')$-collision finder 
$\mathcal{CF}$ against the underlying hash function $H$ and a $(t'',q_s+q_i,\epsilon'')$-forger $\mathcal{F'}$ against the underlying signature scheme $\Psi$, 
where the quantities are related by
$$\epsilon \le \frac{q^2-q}{2^{b/2+1}} + \epsilon' + \epsilon'';\ t \ge max\{t'- (q_s+q_i)t_{sign},t'' \}-q_ht_{op},$$ 
where $q=q_s(n_{max}+1)+q_i$, $q_h=q_sn_{max}+3q_i$, and $t_{sign}$ and $t_{op}$ are, respectively, the maximum running times to perform a signature with $\Psi$ and 
an operation in $\left(\mathbb{F}_{2^{3200}},+\right)$.
\end{theorem}
\end{sloppypar}


\begin{proof}
For the sake of simplicity in the sketch of proof, we only deal with \textit{insert} operations of one block and suppose that a document can have 
a maximum length of $n_{max}$ blocks. 
Let us suppose that there exists a $(t,q_i,q_s,\epsilon)$-forger $\mathcal{F}$ against 
our incremental signature scheme $\mathrm{IncSIG}$.
By interacting adaptively with the challenger in the game defined Section \ref{Unf_game}, the forger eventually outputs a pair
$(D^*,S^*)$. The forger can win the game according to both following possibilities:
\begin{itemize}
 \item case 1: The value $\mu^*$ contained in the successful forgery has already been retrieved in a response to a query.
 \item case 2: The value $\mu^*$ contained in the successful forgery has never been encountered in the responses to the queries.
\end{itemize}

By assuming a forger against our incremental signature scheme $\mathrm{IncSIG}$, we show that the case 1 allows the construction of a collision finder $\mathcal{CF}$ for the hash function $H$ and
the case 2 a forger $\mathcal{F}'$ against the signature scheme $\Psi$. Indeed, we build a collision finder $\mathcal{CF}$ for $H$ and a forger $\mathcal{F}'$ for the underlying 
signature algorithm in the following way:
\begin{itemize}
\begin{sloppypar}
 \item case 1: The collision finder $\mathcal{CF}$ uses $\mathcal{F}$ as a subroutine and simulates an incremental signing oracle as follows: first of all, $\mathcal{CF}$ executes $\Psi.G$ to obtain a 
pair of keys $(sk,pk)$ and conveys $pk$ to the forger $\mathcal{F}$. Whenever $\mathcal{F}$ queries a signature for a document, 
$\mathcal{CF}$ forms correctly the document with an enforced padding to obtain a $n$-block document $D=D_1\|D_2\| \ldots \| D_{n-1}$, then he (she) generates $n$ random blocks 
$R_1,R_2, \ldots, R_n$ and computes 
\begin{eqnarray*}
\mu & = & \sum_{i=1}^{n-1} \mathcal{R}(R_{i}\|R_{i+1}\|D_{i}) \mod{2^{3200}}
\end{eqnarray*}
by resorting to the random oracle for the function $\mathcal{R}$.
After that, he (she) computes a signature $s$ on $\mu$ using the secret key $sk$ and responds to $\mathcal{F}$ with the incremental signature
$S=(R_{1}, R_{2}, \ldots, R_{n}, \mu, s)$. The updating oracle is simulated as follows: whenever $\mathcal{F}$ queries an insertion of one block right after index $i$ in a signed document $(D,R_{1}, R_{2}, \ldots, R_{n}, \mu, s)$,
$\mathcal{CF}$ generates a new random block $R$ and updates the value of $\mu$ (random oracle accesses to $\mathcal{R}$ are needed) 
by first removing the contribution of the pair-wise link $[R_i,R_{i+1}]$ and then adding the contributions of the two new links $[R_i,R]$ and $[R,R_{i+1}]$.
Then he (she) computes a signature $s'$ on the updated hash $\mu'$ and sends the updated incremental signature 
$S'=(R_{1}, R_{2}, \ldots, R_i, R, R_{i+1}, \ldots, R_{n}, \mu', s')$ to $\mathcal{F}$. At the end, $\mathcal{F}$
comes with a new pair $(D^*,S^*)$ where $S^*~=~(R_1^*, \ldots, R_m^*, \mu^*, s^*)$ and $D^*=D_1^* \| \ldots \| D_{m-1}^*$. As we supposed that $s^*$ was obtained from the signing 
oracle $\mathcal{O}^{\Psi.S_{sk}}$, the set $\{R_1^* \| R_2^* \| D_1^*, \ldots , R_{m-1}^* \| R_{m}^* \| D_{m-1}^*\}$ corresponds to a set-collision for the hash function $H$
if a certain condition is fulfilled: the document corresponding to the forged signature is not simply a reordering of the blocks of a queried document.

Let us denote by $q$ the total number of random blocks used as input to the random function such that $q=q_s(n_{max}+1)+q_i$. 
We consider the list $L$ of these $q$ blocks reindexed for the occasion so that $L=(R'_i)_{i=1 \ldots q}$.
Let $\mathbf{AD}$ be the event that the blocks of $L$ are distinct, or in other words, $R'_i\neq R'_j$ for all $i, j \in \llbracket 1,q \rrbracket$ with $i \neq j$. 

When the event $\mathbf{AD}$ occurs the $b$-bit blocks of a signed message can not be permuted without re-evaluating $\mathcal{R}$ to 
the new appearing points, changing the final value of $H$. 
Let us denote simply by \textbf{$\mathrm{S}_{\mathrm{case 1}}$} the event of
success of $\mathcal{F}$ in the current case that we are describing. Then,
\begin{eqnarray*}
\Pr(\mathrm{S}_{\mathrm{case 1}}) & = & \Pr(\mathrm{S}_{\mathrm{case 1}} | \overline{\mathrm{AD}} )\Pr(\overline{\mathrm{AD}}) 
+ \Pr(\mathrm{S}_{\mathrm{case 1}} | \mathrm{AD})\Pr(\mathrm{AD})\\
      & \le & \Pr(\overline{\mathrm{AD}}) + \Pr(\mathrm{S}_{\mathrm{case 1}} | \mathrm{AD})\\
      & \le & \frac{q^2-q}{2^{b/2+1}} + Adv^{H}_ {scr}
\end{eqnarray*}
where $Adv^{H}_ {scr}$ is the advantage of $\mathcal{CF}$ for breaking the hash function $H$ in the sense of set-collision resistance.
 \item case 2: The forger $\mathcal{F'}$ uses $\mathcal{F}$ as a subroutine and simulates its environnement as done by $\mathcal{CF}$ in the case 1 
except for the following differences:
$\mathcal{F'}$ simulates the random oracle for $\mathcal{R}$ by generating on the fly a table mapping input values to random output strings and uses its own signing oracle to 
obtain a signature on a hash value $\mu$.
Eventually, $\mathcal{F}$ comes with a new pair $(D^*,S^*)$ and we have supposed that $s^*$ has not been obtained from the signing oracle, the output
of $\mathcal{F'}$ is setted to $(\mu^*,s^*)$ which corresponds to a valid forgery.
\end{sloppypar}
\end{itemize}
\end{proof}

Now we can focus on the security of $\mathbf{\mathrm{IncSIG^*}}$. A $d$-wise chain staying a $d$-wise chain after an update, the obliviousness property is obvious.
Then we give only the interesting details of the proof of unforgeability concerning the following theorem since this one is very similar to the above. 

\begin{theorem}
\begin{sloppypar}
Let suppose that $\mathcal{F}$ is a $(t,q_i,q_s,\epsilon)$-forger against our incremental signature scheme $\mathrm{IncSIG^*}$, then there exists a $(t',q_h,\epsilon')$-collision finder 
$\mathcal{CF}$ against the hash function $H$ and a $(t'',q_s+q_i,\epsilon'')$-forger $\mathcal{F'}$ against the underlying signature scheme $\Psi$, 
where the quantities are related by
$$\epsilon \le \frac{(q_s+q_i)(n_{max}+1)^2}{2^{(d-1)b/2+1}} + \epsilon' + \epsilon'';$$
$$t \ge max\{t'- (q_s+q_i)t_{sign},t'' \}-q_ht_{op},$$ 
where $q_h=q_sn_{max}+(2d-1)q_i$, and $t_{sign}$ and $t_{op}$ are, respectively, the maximum running times to produce a signature 
with $\Psi$ and to perform an operation in $\left(\mathbb{F}_{2^{3200}},+\right)$.
\end{sloppypar}
\end{theorem}


\begin{proof}
\begin{sloppypar}
Keeping the previous notations,
first notice that in order for the adversary to permute two blocks of a message, say $R_{i}\|R_{i+1}\| \ldots \| R_{i+d-2} \| R_{i+d-1} \| D_i$ 
and $R_{j}\|R_{j+1}\| \ldots \| R_{j+d-2} \| R_{j+d-1} \| D_j$ with $j > i$, there are two possibilities:
($i$) The message is of length two blocks. In this case, $R_{1}\| \ldots \| R_{1+d-3} \| R_{1+d-2}$ must be equal to $R_{2}\| \ldots \| R_{i+d-2} \| R_{i+d-1}$;
($ii$) The message has a length greater than two blocks. In this case, 
having $R_{j+1}\| \ldots \| R_{j+d-2} \| R_{j+d-1}$ equal to $R_{i+1}\| \ldots \| R_{i+d-2} \| R_{i+d-1}$ is a prerequisite. Without this we can not ensure the 
consistency with the $(i+1)$\emph{-th} $d$-tuple in input to $\mathcal{R}$.
\end{sloppypar}
Continuing, it remains to notice that a message of length $n$ contains $n+1$ $(d-1)$-tuples of random blocks.
For each signature or update query providing a signature for a document $D^i$ we consider the list $L^i$ of the $(d-1)$-tuples involved in the produced signature. 
Let $\mathbf{AD'}$ be the event that the elements of $L^i$ are distinct for all $i \in \llbracket1,q_s+q_i\rrbracket$. It follows that:
$$\Pr(\overline{\mathrm{AD'}}) \le \frac{(q_s+q_i)(n_{max}+1)^2}{2^{(d-1)b/2+1}}.$$
\end{proof}

\section{Efficiency}\label{efficiency}

As we can see in Table \ref{tab:IncSIG_Efficiency}, whatever the parametrization used the cost to make a signature does not change. The operations in 
$\mathbb{Z}/2 ^{3200}\mathbb{Z}$ are much more expensive than the hash operations. Therefore, a solution to decrease the number of arithmetic 
operations could be to use larger blocks for the message at the counterpart of less efficient updates. Obviously, such a choice is not interesting 
in incremental cryptography, for which we would prefer to suffer higher costs for the signature generation and perform efficient updates.

When $d$ increases exponentially (and $k$ decreases exponentially), the size of the signature decreases in the same way. For instance, by choosing the 
triple $(b,k,d)=(256,1,256)$ the overhead for the signature size is only $n+255$ bits, that is to say, an overhead of 
approximately $\frac{1}{256}$\emph{-th} the size of the message. On the other hand we notice in Table \ref{tab:IncSIG_Update_Efficiency} that this is accompanied by 
larger update costs, showing that this is a question of compromise.

\begin{table}[!t]
    \centering
    \begin{tabular}{ | p{1.89cm} | p{1.70cm} | p{1.75cm} | p{1.5cm} | }
    \hline
      \centering{Parametrization} &  \centering{Bit-size overhead} & \centering{Number of hash function evaluations} & 
\centering{Number of additions in $\mathbb{Z}/(2 ^{3200}\mathbb{Z})$} \tabularnewline \hline
          \centering{$(k,d)$} & \centering{$nk + (d-1)k$} & \centering{$n$} & \centering{$n-1$} \tabularnewline\hline
    \end{tabular}
    \caption{Efficiency of $\mathrm{IncSIG^*}$: Expansion size of a signature and computational cost of the signing algorithm for a message
of $n$ $b$-bit blocks.}\label{tab:IncSIG_Efficiency}
\end{table}

\begin{table}[!t]
    \centering
    \begin{tabular}{ | p{2.5cm} | p{2.22cm} | p{2.5cm} | }
    \hline
       \centering{Insertion of one block} & \centering{Replacement of one block} & 
\centering{Deletion of one block} \tabularnewline \hline
	  \centering{$(2d-1,d,d-1)$} & \centering{$(2,1,1)$} & \centering{$(2d-1,d-1,d)$}  \tabularnewline \hline
    \end{tabular}
    \caption{Efficiency of $\mathrm{IncSIG^*}$: Computational cost of an update when the parameter $d$ is even, described in the form $(n_h,n_a,n_s)$ where $n_h$ is 
the number of hash evals, $n_a$ and $n_s$ the number of additions and substractions respectively in $\mathbb{Z}/2 ^{3200}\mathbb{Z}$.
}\label{tab:IncSIG_Update_Efficiency}
\end{table}



\begin{sloppypar}
To effectively improve performances, the choice of the underlying primitives is of great importance.
Concerning the underlying randomize function, we need a hash function capable of generating outputs of size 3200 bits. 
The new standard SHA-3, Keccak~\cite{Bertoni09keccakspecifications}, allows the output size to be parametrized. Besides, it can be used with a hash tree mode in order to
increase the degree of parallelism. However, if we can process the input message in parallel, it will be interesting to do the same for the output. 
Nevertheless, Keccak is based on the sponge construction and consequently the blocks which appear in this variable-sized output can not 
be generated in parallel, that is why we could prefer to use a solution based on a counter mode or a GGM technique \cite{DBLP:journals/jacm/GoldreichGM86} to generate them.
A good choice could be the Skein hash function \cite{Ferguson09theskein} which proposes a hash tree mode and permits to generate the output string in parallel as well.
Concerning the underlying signature scheme $\Psi$, one can choose any signature scheme based on the \emph{hash-then-sign} paradigm or a 
signature scheme providing message recovery.
In this latter case, the hash value $\mu$ could be removed from the incremental 
signature since it can be retreived during the verification process.
\end{sloppypar}

\newpage

\section{Conclusion}\label{conclude}

In this paper, we have described a method to construct an incremental asymmetric signature scheme which ensures the perfect privacy property. 
We have shown that we can discard the stronger assumption done about the blocks of the message and still use securely and in a practical way 
an incremental hash function based on pair-wise chaining.
To the best of our knowledge this is the first incremental asymmetric signature whose the update algorithm has a linear time worst case complexity.

Besides, we have shown how we can reduce the size of the signature, 
but at the expense of greater number of hash operations and additions/substractions.
Such a signature scheme is interesting for many applications in which we have to authenticate a lot of documents that continously undergo modifications, 
this is the case of the virus protections, the authentication of files systems and databases. More generally, this is particularly welcome 
for ensuring efficiently a secure handling of files in cloud storage systems.


\bibliographystyle{abbrv}
\bibliography{asia247s-atighehchi}

\begin{thebibliography}{10}

\bibitem{Bellare94incrementalcryptography:}
M.~Bellare, O.~Goldreich, and S.~Goldwasser.
\newblock Incremental cryptography: The case of hashing and signing.
\newblock In {\em CRYPTO}, pages 216--233. Springer, 1994.

\bibitem{Bellare95incrementalcryptography}
M.~Bellare, O.~Goldreich, and S.~Goldwasser.
\newblock Incremental cryptography and application to virus protection.
\newblock In {\em STOC}, pages 45--56. ACM Press, 1995.

\bibitem{Bellare97anew}
M.~Bellare and D.~Micciancio.
\newblock A new paradigm for collision-free hashing: incrementality at reduced
  cost.
\newblock In {\em EUROCRYPT}, pages 163--192. Springer, 1997.

\bibitem{Ber07}
D.~J. Bernstein.
\newblock Better price-performance ratios for generalized birthday attacks.
\newblock In {\em Workshop Record of SHARCS'07: Special-purpose Hardware for
  Attacking Cryptographic Systems}, 2007.

\bibitem{Bertoni09keccakspecifications}
G.~Bertoni, J.~Daemen, M.~Peeters, and G.~V. Assche.
\newblock Keccak specifications, 2009.

\bibitem{conf/asiacrypt/ClarkeDDGS03}
D.~E. Clarke, S.~Devadas, M.~van Dijk, B.~Gassend, and G.~E. Suh.
\newblock Incremental multiset hash functions and their application to memory
  integrity checking.
\newblock In {\em ASIACRYPT}, pages 188--207. Springer, 2003.

\bibitem{DBLP:journals/tcos/ElbazCGLPT09}
R.~Elbaz, D.~Champagne, C.~H. Gebotys, R.~B. Lee, N.~R. Potlapally, and
  L.~Torres.
\newblock Hardware mechanisms for memory authentication: A survey of existing
  techniques and engines.
\newblock {\em Transactions on Computational Science}, 4:1--22, 2009.

\bibitem{Ferguson09theskein}
N.~Ferguson, S.~Lucks, B.~Schneier, D.~Whiting, M.~Bellare, T.~Kohno,
  J.~Callas, and J.~Walker.
\newblock The skein hask function family, 2009.

\bibitem{Computational_complexity_PCIHF}
B.-M. Goi, M.~Siddiqi, and H.-T. Chuah.
\newblock Computational complexity and implementation aspects of the
  incremental hash function.
\newblock {\em IEEE Transactions on Consumer Electronics}, 49:1249--1255, 2003.

\bibitem{DBLP:conf/indocrypt/GoiSC01}
B.-M. Goi, M.~U. Siddiqi, and H.-T. Chuah.
\newblock Incremental hash function based on pair chaining {\&} modular
  arithmetic combining.
\newblock In {\em INDOCRYPT}, pages 50--61. Springer, 2001.

\bibitem{DBLP:journals/jacm/GoldreichGM86}
O.~Goldreich, S.~Goldwasser, and S.~Micali.
\newblock How to construct random functions.
\newblock {\em J. ACM}, 33(4):792--807, 1986.

\bibitem{Goodrich01efficientauthenticated}
M.~T. Goodrich and R.~Tamassia.
\newblock Efficient authenticated dictionaries with skip lists and commutative
  hashing.
\newblock Technical report, TECH. REP., JOHNS HOPKINS INFORMATION SECURITY
  INSTITUTE, 2001.

\bibitem{Micciancio97obliviousdata}
D.~Micciancio.
\newblock Oblivious data structures: Applications to cryptography.
\newblock In {\em STOC}, pages 456--464. ACM Press, 1997.

\bibitem{pairchaining}
R.~C.-W. Phan and D.~Wagner.
\newblock Security considerations for incremental hash functions based on pair
  block chaining.
\newblock {\em Computers \& Security}, pages 131--136, 2006.

\bibitem{Wagner02ageneralized}
D.~Wagner.
\newblock A generalized birthday problem (extended abstract).
\newblock In {\em CRYPTO}, pages 288--303. Springer, 2002.

\bibitem{pair-wise-chaining-hash}
S.~Yunling and M.~Xianghua.
\newblock An overview of incremental hash function based on pair block
  chaining.
\newblock {\em Information Technology and Applications, International Forum
  on}, 3:332--335, 2010.

\end{thebibliography}



\end{document}